\documentclass[12pt]{amsproc}
\usepackage{amsmath,amscd}

\usepackage{epsfig}

\usepackage{amssymb}

\usepackage{latexsym}

\newcommand{\plaq}[5]{

\setlength{\unitlength}{.5in}

\begin{picture}(3.25,2)(0.5,1)

\put(.9,2){$\blacktriangle$}

\put(0.5,2){$#4$}

\put(1,1){\line(0,1){2}}

\put(1,1){\line(1,0){2}}

\put(2,.9){$\blacktriangleright$ }

\put(1.9,0.7){$#1$}

\put(1,3){\line(1,0){2}}

\put(2,2.9){$\blacktriangleright$ }

\put(1.9,3.15){$#3$}

\put(2.9,2){$\blacktriangle$}

\put(3.15,2){$#2$}

\put(3,1){\line(0,1){2}}

\put(1.9,2){$#5$}

\end{picture}
}

\usepackage{hyperref} 

\newtheorem{theorem}{Theorem}[section]
\newtheorem{prop}{Proposition}[section]
\newtheorem{lemma}[theorem]{Lemma}

\newtheorem{definition}[theorem]{Definition}

\theoremstyle{remark}

\numberwithin{equation}{section}

\newcommand{\mpm}{{\mathcal P}M}

\newcommand{\oab}{\omega_{(A,B)}}

\newcommand{\ova}{{\overline{a}}}

\newcommand{\ovA}{{\overline{A}}}

\newcommand{{\tlg}}{\tilde\gamma }

\newcommand{{\tlG}}{\tilde\Gamma }

\newcommand{\pap}{{\mathcal P}_{\ovA}P}

\begin{document}

\title{Parallel Transport  over Path Spaces}

\author{Saikat Chatterjee} \author{Amitabha Lahiri}
\address{{\rm Saikat Chatterjee and Amitabha Lahiri}\\
S.~N.~Bose National Centre for Basic Sciences \\ Block JD,
  Sector III, Salt Lake, Kolkata 700098 \\
  West Bengal, INDIA}
\email{saikat@bose.res.in, amitabha@bose.res.in}

\author{Ambar N. Sengupta}
\address{{\rm Ambar Sengupta}\\ Department of Mathematics,
  Louisiana State University\\  Baton
Rouge, Louisiana 70803\\ USA} \email{sengupta@gmail.com}
\urladdr{http://www.math.lsu.edu/$\sim$sengupta}

\date{15th June, 2010}

 



\begin{abstract} We develop a differential geometric framework for
  parallel transport over path spaces and a corresponding discrete
  theory, an integrated version of the continuum theory, using a
  category-theoretic framework.

\end{abstract}

\maketitle

\section{Introduction}\label{intro}

A considerable body of literature has grown up around the notion of
`surface holonomy', or parallel transport on surfaces, motivated by
the need to have a gauge theory of interaction between charged
string-like objects. Approaches include direct geometric
exploration of the space of paths of a manifold (Cattaneo et al.
\cite{Cat}, for instance), and a very different, category-theory
flavored development (Baez and Schreiber \cite{BS}, for instance).
In the present work we develop both a path-space geometric theory
as well as a category theoretic approach to surface holonomy, and
describe some of the relationships between the two.

As is well known \cite{Baez} from a group-theoretic argument and
also from the fact that there is no canonical ordering of points on
a surface, attempts to construct a group-valued parallel transport
operator for surfaces leads to inconsistencies unless the group is
abelian (or an abelian representation is used). So in our setting,
there are \textit{two} interconnected gauge groups $G$ and $H$. We
work with a fixed principal $G$-bundle $\pi:P\to M$ and connection
$\ovA$; then, viewing the space of $\ovA$-horizontal paths itself
as a bundle over the path space of $M$, we study a particular
type of connection on this path-space bundle which is specified by
means of a second connection $A$ and a field $B$ whose values are
in the Lie algebra $LH$ of $H$. We derive explicit formulas
describing parallel-transport with respect to this connection.  As
far as we are aware, this is the first time an explicit description
for the parallel transport operator has been obtained for a surface
swept out by a path whose endpoints are not pinned. We obtain, in
Theorem \ref{T:reparm}, conditions for the parallel-transport of a
given point in path-space to be independent of the parametrization
of that point, viewed as a path.  We also discuss $H$-valued
connections on the path space of $M$, constructed from the field
$B$. In section \ref{S:CatPlaq} we show how the geometrical data,
including the field $B$, lead to two categories. We prove several
results for these categories and discuss how these categories may
be viewed as `integrated' versions of the differential geometric
theory developed in section \ref{S:Parallel}

In working with spaces of paths one is confronted with the problem
of specifying a differential structure on such spaces. It appears
best to proceed within a simpler formalism. Essentially, one
continues to use terms such as `tangent space' and `differential
form', except that in each case the specific notion is defined
directly (for example, a tangent vector to a space of paths at a
particular path $\gamma$ is a vector field along $\gamma$) rather
than by appeal to a general theory. Indeed, there is a good variety
of choices for general frameworks in this philosophy (see, for
instance, Stacey \cite{Stacey} and Viro \cite{Viro}).  For this
reason we shall make no attempt to build a manifold structure on
any space of paths.

\vskip .25in

\underline{\em Background and Motivation}

Let us briefly discuss the physical background and motivation for
this study. Traditional gauge fields govern interaction between
point particles. Such a gauge field is, mathematically, a
connection $A$ on a bundle over spacetime, with the structure group
of the bundle being the relevant internal symmetry group of the
particle species. The amplitude of the interaction, along some path
$\gamma$ connecting the point particles, is often obtained from the
particle wave functions $\psi$ coupled together using quantities
involving the path-ordered exponential integral ${\mathcal
  P}\exp(-\int_\gamma\ovA)$, which is the same as the
parallel-transport along the path $\gamma$ by the connection
$\ovA$.  If we now change our point of view concerning particles,
and assume that they are extended string-like entities, then each
particle should be viewed not as a point entity but rather a path
(segment) in spacetime. Thus, instead of the two particles located
at two points, we now have two paths $\gamma_1$ and $\gamma_2$; in
place of a path connecting the two point particles we now have a
parametrized path of paths, in other words a surface $\Gamma$,
connecting $\gamma_1$ with $\gamma_2$. The interaction amplitudes
would, one may expect, involve both the gauge field $A$, as
expressed through the parallel transports along $\gamma_1$ and
$\gamma_2$, and an interaction between these two parallel transport
fields. This higher order, or higher dimensional interaction, could
be described by means of a gauge field at the higher level: it
would be a gauge field over the space of paths in spacetime.

  \begin{figure}
   \begin{center}
    \includegraphics[width=2.22in]{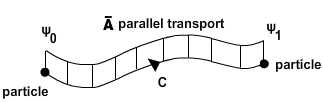}
    
\end{center}
\caption{Point particles interacting via a gauge field}
    \label{fig:pointinteract}
   \end{figure}
   
   \medskip
 
\underline{\em Comparison with other works}

\smallskip

The approach to higher gauge theory developed and explored by Baez
\cite{Baez}, Baez and Schreiber \cite{BS,BS2}, and Lahiri
\cite{Lahiri}, and others cited in these papers, involves an
abstract category theoretic framework of 2-connections and
2-bundles, which are higher-dimensional analogs of bundles and
connections. There is also the framework of gerbes (Chatterjee
\cite{Chat}, Breen and Messing \cite{Breen:2001ie}, Murray
\cite{Murr}).

We develop both a differential geometric framework and
category-theoretic structures. We prove in Theorem \ref{T:reparm}
that a requirement of parametrization invariance imposes a
constraint on a quantity called the `fake curvature' which has been
observed in a related but more abstract context by Baez and
Schreiber \cite[Theorem 23]{BS}.  Our differential geometric
approach is close to the works of Cattaneo et al. \cite{Cat},
Pfeiffer \cite{Pfeiffer}, and Girelli and Pfeiffer \cite{GP}.
However, we develop, in addition to the differential geometric
aspects, the integrated version in terms of categories of diagrams,
an aspect not addressed in \cite{Cat}; also, it should be noted
that our connection form is different from the one used in
\cite{Cat}.  To link up with the integrated theory it is essential
to explore the effect of the $LH$-valued field $B$. To this end we
determine a `bi-holonomy' associated to a path of paths (Theorem
\ref{T:gb}) in terms of the field $B$; this aspect of the theory is
not studied in \cite{Cat} or other works.

Our approach has the following special features:
 \begin{itemize}
 \item we develop the theory with two connections $A$ and $\ovA$ as
   well as a $2$-form $B$ (with the connection $\ovA$ used for
   parallel-transport along any given string-like object, and the
   forms $A$ and $B$ used to construct parallel-transports between
   different strings);
   
 \item we determine, in Theorem \ref{T:gb}, the `bi-holonomy'
   associated to a path of paths using the $B$-field;
\item we allow `quadrilaterals' rather than simply bigons in the
  category theoretic formulation, corresponding to having strings
  with endpoints free to move rather than fixed-endpoint strings.
\end{itemize}
Our category theoretic considerations are related to notions about
double categories introduced by Ehresmann \cite{Ehr1,Ehr2} and
explored further by Kelly and Street \cite{KS}.
   
  \begin{figure}{ }
 \begin{center}
    \includegraphics[width=2.22in]{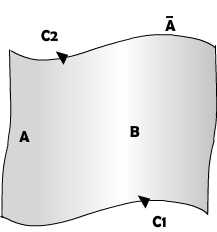}
    
    \caption{Gauge fields along paths $c_1$ and $c_2$ interacting
      across a surface}
    \label{fig:c1c2interact}

 \end{center}
   \end{figure}
 
\section{Connections on Path-space  Bundles}\label{S:Parallel}

In this section we will construct connections and
parallel-transport for a pair of intertwined structures: path-space
bundles with structure groups $G$ and $H$, which are Lie groups
intertwined as explained below in (\ref{pfid}).  For the physical
motivation, it should be kept in mind that $G$ denotes the gauge
group for the gauge field along each path, or string, while $H$
governs, along with $G$, the interaction between the gauge fields
along different paths.

An important distinction between existing differential geometric
approaches (such as Cattaneo et al. \cite{Cat}) and the `integrated
theory' encoded in the category-theoretic framework is that the
latter necessarily involves two gauge groups: a group $G$ for
parallel transport along paths, and another group $H$ for parallel
transport between paths (in path space). We shall develop the
differential geometric framework using a pair of groups $(G,H)$ so
as to be consistent with the `integrated' theory.  Along with the
groups $G$ and $H$, we use a fixed smooth homomorphism $\tau:H\to
G$ and a smooth map
$$G\times H\to
H:(g,h)\mapsto\alpha(g)h$$ such that each $\alpha(g)$ is an
automorphism of $H$, such that the identities
\begin{equation}\label{pfid}
\begin{split}
  \tau\bigl(\alpha(g)h\bigr) &=g\tau(h)g^{-1}\\
  \alpha\bigl(\tau(h)\bigr)h' &=hh'h^{-1}\end{split}
\end{equation}
hold for all $g\in G$ and $h,h'\in H$. The derivatives $\tau'(e)$
and $\alpha'(e)$ will be denoted simply as $ \tau:LH\to LG$ and $
\alpha:LG\to LH$. (This structure is called a {\em Lie 2-group} in
\cite{Baez,BS}).

To summarize very rapidly, anticipating some of the notions
explained below, we work with a principal $G$-bundle $\pi:P\to M$
over a manifold $M$, equipped with connections $A$ and $\ovA$, and
an $\alpha$-equivariant vertical $2$-form $B$ on $P$ with values in
the Lie algebra $LH$. We then consider the space $\pap$ of
$\ovA$-horizontal paths in $P$, which forms a principal $G$-bundle
over the path-space ${\mathcal P}M$ in $M$. Then there is an
associated vector bundle $E$ over ${\mathcal P}M$ with fiber $LH$;
using the $2$-form $B$ and the connection form $\ovA$ we construct,
for any section $\sigma$ of the bundle $P\to M$, an $LH$-valued
$1$-form $\theta^{\sigma}$ on ${\mathcal P}M$.  This being a
connection over the path-space in $M$ with structure group $H$,
parallel-transport by this connection associates elements of $H$ to
{\em parametrized surfaces} in $M$.  Most of our work is devoted to
studying a second connection form $\oab$, which is a connection on
the bundle $\pap$ which we construct using a second connection $A$
on $P$.  Parallel-transport by $\oab$ is related to
parallel-transport by the $LH$-valued connection form
$\theta^{\sigma}$.

\bigskip


\underline{\em Principal bundle and the connection $\ovA$}

\smallskip

Consider a principal $G$-bundle
$$\pi:P\to M$$ with the right-action of the Lie group $G$ on $P$
denoted
$$P\times G\to P: (p,g)\mapsto pg=R_gp.$$
Let $\ovA$ be a connection on this bundle. The space ${\mathcal
  P}_{\ovA}P$ of $\ovA$-horizontal paths in $P$ may be viewed as a
principal $G$-bundle over ${\mathcal P}M$, the space of smooth
paths in $M$.

We will use the notation $pK\in T_pP$, for any point $p\in P$ and
Lie-algebra element $K\in LG$, defined by
$$pK=\frac{d}{dt}\Big|_{t=0}p\cdot\exp(tK).$$
It will be convenient to keep in mind that we always use $t$ to
denote the parameter for a path on the base manifold $M$ or in the
bundle space $P$; we use the letter $s$ to parametrize a path in
path-space.

 \vfill\eject

\noindent\underline{\em The tangent space to ${\mathcal
    P}_{\ovA}P$}  
\smallskip

The points of the space ${\mathcal P}_{\ovA}P$ are
$\ovA$-horizontal paths in $P$. Although we call ${\mathcal
  P}_{\ovA}P$ a `space' we do not discuss any topology or manifold
structure on it. However, it is useful to introduce certain
differential geometric notions such as tangent spaces on ${\mathcal
  P}_{\ovA}P$. It is intuitively clear that a tangent vector at a
`point' ${\tilde\gamma} \in {\mathcal P}_{\ovA}P$ ought to be a
vector field on the path ${\tilde\gamma}$.  We formalize this idea
here (as has been done elsewhere as well, such as in Cattaneo et
al. \cite{Cat}).

If ${\mathcal P}X$ is a space of paths on a manifold $X$, we denote
by ${\rm ev}_t$ the evaluation map
\begin{equation}\label{E:defevt}
{\rm ev}_t:{\mathcal P}X\to X: \gamma\mapsto {\rm
  ev}_t(\gamma)=\gamma(t). 
\end{equation}

Our first step is to understand the tangent spaces to the bundle
${\mathcal P}_{\ovA}P$. The following result is preparation for the
definition (see also \cite[Theorem 2.1]{Cat}).

\begin{prop}\label{P:tngt} Let $\ovA$ be a connection on a
  principal 
  $G$-bundle $\pi:P\to M$, and
$${{\tlG}}:[0,1]\times[0,1]\to
P:(t,s)\mapsto{\tlG}(t,s)={\tlG}_s(t)$$ a 
smooth map, and
 $${\tilde v}_s(t)=\partial_s{\tlG}(t,s).$$
 Then the following are equivalent: \begin{itemize}
\item[\rm (i)] Each transverse  path 
$${{\tlG}}_s:[0,1]\to P:
t\mapsto {\tlG}(t,s)$$ is $\ovA$-horizontal.
\item[\rm (ii)] The initial path ${\tlG}_0$ is $\ovA$-horizontal,
and the `tangency condition'
\begin{equation}\label{tngtcond}
\begin{frac}{\partial\ovA({\tilde
v}_s(t))}{\partial t}\end{frac}=
F^{\ovA}\left({\partial_t}{\tilde\Gamma}(t,s), 
{\tilde
v}_s(t)\right)
\end{equation}
holds, and thus also
\begin{equation}\label{duh}
\ovA\bigl({\tilde v}_s(T)\bigr) - \ovA\bigl({\tilde v}_s(0) \bigr)
= \int_0^TF^{\ovA}\left({\partial_t}{\tilde\Gamma}(t,s),
{\tilde v}_s(t)\right)\, dt,
\end{equation}
for every $T,s\in [0,1]$.
\end{itemize}

\end{prop}

Equation (\ref{tngtcond}), and variations on it, is sometimes
referred to as the Duhamel formula and sometimes a `non-abelian
Stokes formula.' We can write it more compactly by using the notion
of a Chen integral. With suitable regularity assumptions, a
$2$-form $\Theta$ on a space $X$ yields a $1$-form, denoted
$\int\Theta$, on the space ${\mathcal P}X$ of smooth paths in $X$;
if $c$ is such a path, a `tangent vector' $v\in T_c({\mathcal P}X)$
is a vector field $t\mapsto v(t)$ along $c$, and the evaluation of
the $1$-form $\int\Theta$ on $v$ is defined to be
\begin{equation}\label{E:chentheta}
\left(\int\Theta\right)_cv=\left(\int_c\Theta\right)(v)=\int_0^1\Theta
\bigl(c'(t),v(t)\bigr)\,dt. 
\end{equation}
The $1$-form $\int\Theta$, or its localization to the tangent space 
$T_c({\mathcal P}X)$, is called the Chen integral of $\Theta$.
Returning to our context, we then have
\begin{equation}\label{nonabstokeschen}
{\rm ev}_T^*\ovA- {\rm ev}_0^*\ovA
=\int_0^TF^{\ovA},
\end{equation}
where the integral on the right is a Chen integral; here it is, by
definition, the $1$-form on $\pap$ whose value on a vector ${\tilde
  v}_s\in T_{\tlG_s}\pap$ is given by the right side of
(\ref{tngtcond}). The pullback ${\rm ev}_t^*\ovA$ has the obvious
meaning.

\begin{proof}  From the definition of the curvature form
  $F^{\ovA}$, 
we have
$$F^{\ovA}\bigl(\partial_t{\tlG},
\partial_s{\tlG})=\partial_t\bigl(\ovA(\partial_s{\tlG})\bigr)
-\partial_s\bigl(\ovA(\partial_t{\tlG})\bigr)-
\ovA\bigl(\underbrace{[\partial_t{\tlG},\partial_s
  {\tlG}]}_{0}\bigr)+\left[\ovA(\partial_t{\tlG}),
  \ovA(\partial_s{\tlG})\right].
$$
So
\begin{equation}\label{dtds}\begin{split}
\partial_t\bigl(\ovA(\partial_s{\tlG})\bigr)-
F^{\ovA}(\partial_t{\tlG},\partial_s{\tlG})&=
\partial_s\bigl(\ovA(\partial_t{\tlG})\bigr)
-\bigl[\ovA(\partial_t{\tlG}),\ovA(\partial_s{\tlG})
\bigr]\\
&= 0\qquad\hbox{if $\ovA(\partial_t{\tlG})=0$,}
\end{split}
\end{equation}
thus proving  (\ref{tngtcond}) if (i) holds. The equation
(\ref{duh}) then follows by integration.

Next suppose (ii) holds. Then, from the first line in (\ref{dtds}), 
we have
\begin{equation}\label{pts}
\partial_s\bigl(\ovA(\partial_t{\tlG})\bigr)-
\bigl[\ovA(\partial_t{\tlG}),\ovA(\partial_s{\tlG})\bigr]=0.
\end{equation}
Now let $s\mapsto h(s)\in G$ describe parallel-transport along
$s\mapsto{\tlG}(s,t)$; then
$$h'(s)h(s)^{-1}=-
\ovA\bigl(\partial_s{\tlG}(s,t)\bigr),\quad\hbox{and $h(0)=e$.}$$ 
Then
\begin{equation}\begin{split}
\partial_s\left(h(s)^{-1}\ovA\bigl(\partial_t{\tlG}(t,s)
  \bigr)h(s)\right)  
&\\
&\hskip -1in={\rm Ad}\bigl(h(s)^{-1}\bigr)\left[
\partial_s\bigl(\ovA(\partial_t{\tlG})\bigr)
-\bigl[\ovA(\partial_t{\tlG}), \ovA(\partial_s{\tlG})\right] 
\end{split}\end{equation}
and the right side here is $0$, as seen in (\ref{pts}). Therefore,
$$h(s)^{-1}\ovA\bigl(\partial_t{\tlG}(t,s)\bigr)h(s)$$
is independent of $s$, and hence is equal to its value at $s=0$.
Thus, if $\ovA$ vanishes on $\partial_t{\tlG}(t,0)$ then it also
vanishes in $\partial_t{\tlG}(t,s)$ for all $s\in [0,1]$. In
conclusion, if the initial path ${\tlG}_0$ is $\ovA$-horizontal,
and the tangency condition (\ref{tngtcond}) holds, then each
transverse path ${\tlG}_s$ is $\ovA$-horizontal.
\end{proof}

In view of the preceding result, it is natural to define the
tangent spaces to $ {\mathcal P}_{\ovA}P$ as follows:

\begin{definition}
  The tangent space to $ {\mathcal P}_{\ovA}P$ at ${\tilde\gamma}$ 
  is the linear space of all vector fields $t\mapsto {\tilde
    v}(t)\in T_{{\tilde\gamma}(t)}P$ along $\tilde\gamma$ for which 
\begin{equation}\label{def:vTPAP}
\setlength{\fboxrule}{0.5pt}\framebox{\parbox{2in}{$
\begin{frac}{\partial\ovA({\tilde
v}(t))}{\partial t}\end{frac}- F^{\ovA}\left({\tilde\gamma}'(t),
{\tilde v}(t)\right)=0$}}
\end{equation}
holds for all $t\in [0,1]$. 
\end{definition}

The {\em vertical subspace} in $T_{{\tilde\gamma}}{\mathcal
  P}_{\ovA}P$ consists of all vectors ${\tilde v}(\cdot)$ for which  
${\tilde v}(t)$ is vertical in $T_{{\tilde\gamma}(t)}P$ for every
$t\in [0,1]$.

Let us note one consequence:
\begin{lemma}\label{l:lift}  Suppose $\gamma:[0,1]\to M$ is a
  smooth path, and ${\tilde\gamma}$ an $\ovA$-horizontal lift. Let
  $v:[0,1]\to TM$ be a vector field along $\gamma$, and ${\tilde
    v}(0)$ any vector in $ T_{{\tilde\gamma}(0)}P$ with
  $\pi_*{\tilde v}(0)=v(0)$. Then there is a unique vector field
  ${\tilde v}\in T_{{\tilde\gamma}}{\mathcal P}_{\ovA}P$ whose
  projection down to $M$ is the vector field $v$, and whose initial
  value is ${\tilde v}(0)$.
\end{lemma}
\begin{proof} The first-order differential equation
  (\ref{def:vTPAP}) determines the vertical part of ${\tilde
    v}(t)$, from the initial value. Thus ${\tilde v}(t)$ is this
  vertical part plus the $\ovA$-horizontal lift of $v(t)$ to
  $T_{{\tilde\gamma}(t)}P$.
\end{proof}

\smallskip

\noindent\underline{\em Connections induced from $B$} 
\smallskip

All through our work, $B$ will denote a vertical
$\alpha$-equivariant $2$-form on $P$ with values in $LH$. In more
detail, this means that $B$ is an $LH$-valued $2$-form on $P$ which
is vertical in the sense that
$$B(u,v)=0\quad\hbox{if $u$ or $v$ is vertical,}$$
and $\alpha$-equivariant in the sense that
$$R_g^*B=\alpha(g^{-1})B\quad\hbox{for all $g\in G$}$$
wherein $R_g:P\to P:p\mapsto pg$ is the right action of $G$ on the
principal bundle space $P$,
and $$\alpha(g^{-1})B=d\alpha(g^{-1})|_eB,$$ recalling that
$\alpha(g^{-1})$ is an automorphism $H\to H$.

Consider an $\ovA$-horizontal ${\tilde\gamma}\in\pap$, and a smooth
vector field $X$ along $\gamma=\pi\circ\tlg$; take any lift
${\tilde X}_{\tlg}$ of $X$ along $\tlg$, and set
\begin{equation}\label{E:defthetatlg}
{\theta}_{\tlg}(X)\stackrel{\rm def}{=}
\left(\int_{\tlg}B\right)({\tilde X}_{\tlg})= 
\int_0^1B\bigl({\tilde\gamma}'(u), {\tilde
  X}_{\tlg}(u)\bigr)\,du.
\end{equation}  
This is independent of the choice of ${\tilde X}_{\tlg}$ (as any
two choices differ by a vertical vector on which $B$ vanishes) and
specifies a linear form ${\theta}_{\tlg}$ on $T_{\gamma}({\mathcal
  P}M)$ with values in $LH$.  If we choose a different horizontal
lift of $\gamma$,  a path ${\tilde\gamma}g$, with $g\in G$,
then
\begin{equation}\label{E:thetag}
{\theta}_{{\tlg}g}(X)=\alpha(g^{-1})\theta_{\tlg}(X).
\end{equation}
Thus, one may view ${\tilde\theta}$ to be a $1$-form on $\mpm$ with
values in the vector bundle $E\to \mpm$ associated to $\pap\to\mpm$
by the action $\alpha$ of $G$ on $LH$.

Now fix a section $\sigma:M\to P$, and for any path $\gamma\in
\mpm$ let ${\tilde\sigma}(\gamma)\in\pap$ be the $\ovA$-horizontal
lift with initial point $\sigma\bigl(\gamma(0)\bigr)$.  Thus,
${\tilde\sigma}:\mpm\to\pap$ is a section of the bundle
$\pap\to\mpm$. Then we have the $1$-form $\theta^{\sigma}$ on
$\mpm$ with values in $LH$ given as follows: for any $X\in
T_{\gamma}(\mpm)$,
\begin{equation}\label{E:defthetsigX}
(\theta^\sigma)(X)=\theta_{{\tilde\sigma}(\gamma)}(X).
\end{equation}
We shall view $\theta^{\sigma}$ as a connection form for the
trivial $H$-bundle over $\mpm$. Of course, it depends on the
section $\sigma$ of $\pap\to\mpm$, but in a `controlled' manner,
i.e., the behavior of $\theta^{\sigma}$ under change of $\sigma$ is
obtained using (\ref{E:thetag}).

\smallskip

\noindent\underline{\em Constructing the connection
  $\omega_{(A,B)}$} 
\smallskip

Our next objective is to construct connection forms on ${\mathcal
  P}_{\ovA}P$. To this end, fix a connection $A$ on $P$, in
addition to the connection $\ovA$ and the $\alpha$-equivariant
vertical $LH$-valued $2$-form $B$ on $P$.

The evaluation map
 at any time $t\in [0,1]$, given by
 $${\rm ev}_t:{\mathcal P}_{\ovA}P\to P: {\tilde\gamma}\mapsto
 {\tilde\gamma}(t),$$ commutes with the projections ${\mathcal
   P}_{\ovA}P\to {\mathcal P}M$ and $P\to M$, and the evaluation
 map ${\mathcal P}M\to M$.  We can pull back any connection $A$ on
 the bundle $P$ to a connection ${\rm ev}_t^*A$ on ${\mathcal
   P}_{\ovA}P$.

 Given a $2$-form $B$ as discussed above, consider the $LH$-valued
 $1$-form $Z$ on ${\mathcal P}_{\ovA}P$ specified as follows. Its
 value on a vector ${\tilde v}\in T_{{\tilde\gamma}}{\mathcal
   P}_{\ovA}P$ is defined to be
\begin{equation}\label{def:Zv}
Z({\tilde v})=  \int_{0}^1 B  \left({\tilde\gamma}'(t), {\tilde
v}(t)\right)\,dt.
\end{equation}
 Thus
\begin{equation}\label{def:ZAB} Z=
 \int_{0}^1 B,
\end{equation}
where on the right we have the Chen integral (discussed earlier in
(\ref{E:chentheta})) of the $2$-form $B$ on $P$, lifting it to an
$LH$-valued $1$-form on the space of ($\ovA$-horizontal) smooth
paths $[0,1]\to P$.  The Chen integral here is, by definition, the
$1$-form on ${\mathcal P}_{\ovA}P$ given by
$$ 
 {\tilde v}\in T_{{\tilde\gamma}}{\mathcal
P}_{\ovA}P\mapsto \int_{0}^1B\left({\tilde\gamma}'(t), {\tilde
v}(t)\right)\,dt.$$
Note that $Z$ and the form $\theta$ are closely related:
\begin{equation}\label{E:Ztidev}
Z({\tilde v})=\theta_{\tlg}(\pi_*{\tilde v}).\end{equation}
 Now define the $1$-form $\omega_{(A,B)}$ by
\begin{equation}\label{def:omegaAB}
\setlength{\fboxrule}{0.5pt}\framebox{\parbox{2.5in}{$
\qquad\omega_{(A,B)}= {\rm ev}_1^*A +{\tau}(Z)$}}
\end{equation}
Recall that $\tau:H\to G$ is a homomorphism, and, for any $X\in
LH$, we are writing $\tau(X)$ to mean $\tau'(e)X$; here
$\tau'(e):LH\to LG$ is the derivative of $\tau$ at the
identity. The utility of bringing in $\tau$ becomes clear only when
connecting these developments to the category theoretic formulation
of section 3. A similar construction, but using only one algebra
$LG\,,$ is described by Cattaneo et al.~\cite{Cat}. However, as we
pointed out earlier, a parallel transport operator for a surface
cannot be constructed using a single group unless the group is
abelian. To allow non-abelian groups, we need to have two groups
intertwined in the structure described in (\ref{pfid}), and thus we
need $\tau$.

Note that $\omega_{(A,B)}$ is simply the connection ${\rm ev}_1^*A$
on the bundle ${\mathcal P}_{\ovA}P$, shifted by the $1$-form
$\tau(Z)$. In the finite-dimensional setting it is a standard fact
that such a shift, by an equivariant form which vanishes on
verticals, produces another connection; however, given that our
setting is, technically, not identical to the finite-dimensional
one, we shall prove this below in Proposition \ref{P:omAB}.

 Thus,
\begin{equation}\label{def:omegaABv} \omega_{(A,B)}({\tilde v})
  =  A\bigl({\tilde v}(1)\bigr)
  +\int_0^1{\tau}B\bigl({\tilde\gamma}'(t), {\tilde
    v}(t)\bigr)\,dt.
\end{equation} 
We can rewrite this as
\begin{equation}\label{def:omegaABv2} 
      \omega_{(A,B)}
      =  {\rm ev}_0^*A +\left[{\rm ev}_1^*(A-\ovA)-{\rm
          ev}_0^*(A-\ovA)\right] +\int_0^1\bigl(F^{\ovA}+{\tau}B
      \bigr).  
\end{equation} 
To obtain this we have simply used the relation (\ref{duh}). The
advantage in (\ref{def:omegaABv2}) is that it separates off the end
point terms and expresses $\oab$ as a perturbation of the simple
connection ${\rm ev}_0^*A$ by a vector in the tangent space
$T_{{\rm ev}_0^*A}{\mathcal A},$ where ${\mathcal A}$ is the space
of connections on the bundle $\pap$. Here note that the `tangent
vectors' to the affine space $\mathcal A$ at a connection $\omega$
are the $1$-forms $\omega_1-\omega$, with $\omega_1$ running over
$\mathcal A$.  A difference such as $\omega_1-\omega$ is precisely
an equivariant $LG$-valued $1$-form which vanishes on vertical
vectors.

Recall that the group $G$ acts on $P$ on the right
$$P\times G\to P: (p,g)\mapsto R_gp=pg$$ and this induces a natural
right action of $G$ on ${\mathcal P}_{\ovA}P$:
$${\mathcal
  P}_{\ovA}P\times G\to {\mathcal P}_{\ovA}P:
({\tilde\gamma},g)\mapsto R_g{\tilde\gamma}={\tilde\gamma}g$$ Then
for any vector $X$ in the Lie algebra $LG$, we have a vertical
vector
$${\tilde X}({{\tilde\gamma}})\in T_{{\tilde\gamma}}{\mathcal
P}_{\ovA}P$$ given by
$${\tilde
X}({{\tilde\gamma}})(t)=\begin{frac}{d}{du}
\end{frac}\Big|_{u=0}{\tilde\gamma}(t)\exp(uX)$$

\begin{prop}\label{P:omAB} The form $\omega_{(A,B)}$ is a
  connection form on the principal $G$-bundle $\pap\to \mpm$. More
  precisely,
$$\omega_{(A,B)}\bigl((R_g)_*v\bigr)=
{\rm Ad}(g^{-1})\omega_{(A,B)}(v)$$ for every $g\in G$, ${\tilde
  v}\in T_{\tilde\gamma}\bigl(\pap\bigr)$ and
$$\omega_{(A,B)}({\tilde X})=X$$ for every $X\in LG$.
\end{prop}
 \begin{proof} It will suffice to show that for every $g\in G$,
$$Z\bigl((R_g)_*v\bigr)={\rm Ad}(g^{-1})Z(v)$$
  and every vector $v$ tangent to ${\mathcal
P}_{\ovA}P$, and
$$Z({\tilde X})=0$$ for every $X\in LG$. 

From (\ref{def:ZAB}) and the fact that $B$ vanishes on verticals it
is clear that $Z({\tilde X})$ is $0$. The equivariance under the
$G$-action follows also from (\ref{def:ZAB}), on using the
$G$-equivariance of the connection form $A$ and of the $2$-form
$B$, and the fact that the right action of $G$ carries
$\ovA$-horizontal paths into $\ovA$-horizontal paths.
\end{proof}

\smallskip
\underline{\em Parallel transport by $\omega_{(A,B)}$}

Let us examine how a path is parallel-transported by
$\omega_{(A,B)}$. At the infinitesimal level, all we need is to be
able to lift a given vector field $v:[0,1]\to TM$, along
$\gamma\in{\mathcal P}M$, to a vector field ${\tilde v}$ along
$\tilde\gamma$ such that:
\begin{itemize}
\item[(i)] ${\tilde v}$ is a vector in
  $T_{{\tilde\gamma}}\bigl(\pap\bigr)$, 
which means that it satisfies the   equation (\ref{def:vTPAP}):
\begin{equation}\label{davt}
\begin{frac}{\partial\ovA({\tilde
v}(t))}{\partial t}\end{frac}=
F^{\ovA}\left({\tilde\gamma}'(t),{\tilde 
v}(t)\right);
\end{equation}
\item[(ii)] ${\tilde v}$ is $\omega_{(A,B)}$-horizontal, i.e.
satisfies
the  equation
\begin{equation}\label{omabzero}
A\bigl({\tilde v}(1)\bigr)
+\int_0^1{\tau}B\bigl({\tilde\gamma}'(t), 
{\tilde v}(t)\bigr)\,dt=0.
\end{equation}
\end{itemize}
The following result gives a constructive description of ${\tilde v}$.

\begin{prop}\label{P:parlgam}  Assume that $A$, $\ovA$, $B$, and
  $\omega_{(A,B)}$ are as specified before. Let
  ${\tilde\gamma}\in\pap$, and
  $\gamma=\pi\circ{\tilde\gamma}\in\mpm$ its projection to a path
  on $M$, and consider any $v\in T_{\gamma}\mpm$. Then the
  $\omega_{(A,B)}$-horizontal lift ${\tilde v}\in
  T_{{\tilde\gamma}}\pap$ is given by
$${\tilde v}(t)={\tilde v}_{\ovA}^{\rm h}(t)+{\tilde v}^{\rm
  v}(t),$$ where ${\tilde v}_{\ovA}^{\rm h}(t)\in
T_{{\tilde\gamma}(t)}P$ is the $\ovA$-horizontal lift of $v(t)\in
T_{\gamma(t)}M$, and
\begin{equation}\label{tlvt}
{\tilde v}^{\rm v}(t) ={\tilde\gamma}(t)\left[\ovA({\tilde
v}(1)\bigr) -\int_t^1F^{\ovA}\bigl(
{\tilde\gamma}'(u),{\tilde
  v}^h_{\ovA}(u)\bigr)\,du\right]
\end{equation}
 wherein 
\begin{equation}\label{tv02}{\tilde v}(1)={\tilde v}^h_A(1)+
{\tilde\gamma}(1)X,
\end{equation}
with ${\tilde v}^h_A(1)$ being the $A$-horizontal lift of $v(1)$ in
$T_{{\tilde\gamma}(1)}P$, and 
\begin{equation}\label{Atv02}
X =- \int_0^1 {\tau}B \bigl({\tilde\gamma}'(t), {\tilde
v}^h_{\ovA}(t)\bigr)\,dt.
\end{equation}
\end{prop}
Note that $X$ in (\ref{Atv02}) is $A\bigl({\tilde v}(1)\bigr)$.

Note also that since $\tilde v$ is tangent to $\pap$, the vector
${\tilde v}^{\rm v}(t) $ is also given by 
\begin{equation}\label{vvt2}
{\tilde v}^{\rm v}(t) ={\tilde\gamma}(t)\left[{\ovA}\bigl({\tilde
    v}(0)\bigr)+\int_0^tF^{\ovA}\bigl( {\tilde\gamma}'(u),{\tilde
    v}^h_{\ovA}(u)\bigr)\,du\right] 
\end{equation}

\begin{proof} The $\oab$ horizontal lift $\tilde v$ of $v$ in
  $T_{{\tilde\gamma}}\bigl(\pap\bigr)$ is the vector field $\tilde
  v$ along ${\tilde\gamma}$ which projects by $\pi_*$ to $v$ and
  satisfies the condition (\ref{omabzero}):
\begin{equation}\label{omabzero2}
A\bigl({\tilde v}(1)\bigr)
+\int_0^1{\tau}B\bigl({\tilde\gamma}'(t), 
{\tilde v}(t)\bigr)\,dt=0.
\end{equation}
Now for each $t\in [0,1]$, we can split the vector ${\tilde v}(t)$
into an $\ovA$-horizontal part and a vertical part ${\tilde
  v}^v(t)$ which is essentially the element $\ovA\bigl({\tilde
  v}^v(t)\bigr)\in LG$ viewed as a vector in the vertical subspace
in $T_{{\tilde\gamma}(t)}P$:
$${\tilde v}(t)={\tilde v}^h_{\ovA}(t)+{\tilde v}^v(t)$$
and the vertical part here is given by
$${\tilde v}^v(t)={\tilde\gamma}(t) \ovA\bigl({\tilde
  v}(t)\bigr).$$  
Since the vector field ${\tilde v}$ is actually a vector in
$T_{{\tilde\gamma}}\bigl(\pap\bigr)$, we have, from (\ref{davt}),  the relation 
$$\ovA\bigl({\tilde v}(t)\bigr)=\ovA({\tilde
  v}(1)\bigr) -\int_t^1F^{\ovA}\bigl( {\tilde\gamma}'(u),{\tilde
  v}^h_{\ovA}(u)\bigr)\,du.$$ We need now only verify the
expression (\ref{tv02}) for ${\tilde v}(1)$. To this end, we first
split this into $A$-horizontal and a corresponding vertical part:
$${\tilde
v}(1)={\tilde
v}^h_A(1) +{\tilde\gamma}(1)A\bigl({\tilde
v}(1)\bigr)$$
The vector $A\bigl({\tilde
v}(1)\bigr)$ is obtained from  (\ref{omabzero2}), and thus proves
(\ref{tv02}). 
\end{proof}

There is an observation to be made from Proposition
\ref{P:parlgam}. The equation (\ref{Atv02}) has, on the right side,
the integral over the entire curve $\tilde\gamma$. Thus, if we were
to consider parallel-transport of only, say, the `left half' of
$\tilde\gamma$, we would, in general, end up with a {\em different}
path of paths!

\medskip

\underline{\em Reparametrization Invariance}
\medskip

If a path is reparametrized, then, technically, it is a different
point in path space. Does parallel-transport along a path of paths
depend on the specific parametrization of the paths? We shall
obtain conditions to ensure that there is no such
dependence. Moreover, in this case, we shall also show that
parallel transport by $\oab$ along a path of paths depends
essentially on the surface swept out by this path of paths, rather
than the specific parametrization of this surface.

For the following result, recall that we are working with Lie
groups $G$, $H$, smooth homomorphism $\tau:H\to G$, smooth map
$\alpha:G\times H\to H:(g,h)\mapsto \alpha(g)h$, where each
$\alpha(g)$ is an automorphism of $H$, and the maps $\tau$ and
$\alpha$ satisfy (\ref{pfid}). Let $\pi:P\to M$ be a principal
$G$-bundle, with connections $A$ and $\ovA$, and $B$ an $LH$-valued
$\alpha$-equivariant $2$-form on $P$ vanishing on vertical
vectors. As before, on the space $\pap$ of $\ovA$-horizontal paths,
viewed as a principal $G$-bundle over the space $\mpm$ of smooth
paths in $M$, there is the connection form $\oab$ given by
 $$\oab={\rm ev}_1^*A+\int_0^1{\tau}B.$$
 By a `smooth path' $s\mapsto\Gamma_s$ in $\mpm$ we mean a smooth
 map $$[0,1]^2\to M: (t,s)\mapsto\Gamma(t,s)=\Gamma_s(t),$$ viewed
 as a path of paths $\Gamma_s\in\mpm$.

  With this notation and framework, we have:

  \begin{theorem}\label{T:reparm}
    Let $$\Phi:[0,1]^2\to[0,1]^2:(t,s)\mapsto(\Phi_s(t),\Phi^t(s)
    )$$ be a smooth diffeomorphism which fixes each vertex of
    $[0,1]^2$.  Assume that
  \begin{itemize}
  \item[{\rm (i)}] either
  \begin{equation}\label{fakecurvzero}
    F^{\ovA}+\tau(B)=0\end{equation}
  and $\Phi$  carries each   $s$-fixed section $[0,1]\times\{s\}$
  into an $s$-fixed section $[0,1]\times\{\Phi^0(s)\}$; 
  \item[{\rm (ii)}] or
\begin{equation}\label{modiffake}
  \left[{\rm ev}_1^*(A-\ovA)-{\rm ev}_0^*(A-\ovA)\right]
  +\int_0^1\bigl(F^{\ovA}+{\tau}B)=0,
\end{equation} 
 $\Phi$ maps each boundary edge of $[0,1]^2$ into itself, and 
$\Phi^0(s)=\Phi^1(s)$ for all $s\in [0,1]$.  
\end{itemize}
Then the $\oab$-parallel-translate of the point $\tlG_0\circ\Phi_0$
along the path $s\mapsto (\Gamma\circ\Phi)_s$, is
$\tlG_1\circ\Phi_1$, where $\tlG_1$ is the
$\oab$-parallel-translate of $\tlG_0$ along $s\mapsto\Gamma_s$.

As a special case, if the path $s\mapsto\Gamma_s$ is constant and
$\Phi_0$ the identity map on $[0,1]$, so that $\Gamma_1$ is simply
a reparametrization of $\Gamma_0$, then, under conditions (i) or
(ii) above, the $\oab$-parallel-translate of the point $\tlG_0$
along the path $s\mapsto (\Gamma\circ\Phi)_s$, is $\tlG_0
\circ\Phi_1$, i.e., the appropriate reparametrizaton of the
original path $\tlG_0$.

\end{theorem}

Note that the path $(\tlG\circ\Phi)_0$ projects down to
$(\Gamma\circ\Phi)_0$, which, by the boundary behavior of $\Phi$,
is actually that path $\Gamma_0\circ\Phi_0$, in other words
$\Gamma_0$ reparametrized. Similarly, $(\tlG\circ\Phi)_1$ is an
$\ovA$-horizontal lift of the path $\Gamma_1$, reparametrized by
$\Phi_1$.

If $A=\ovA$ then conditions (\ref{modiffake}) and
(\ref{fakecurvzero}) are the same, and so in this case the weaker
condition on $\Phi$ in (ii) suffices.

\begin{proof} Suppose (\ref{fakecurvzero}) holds. Then the
  connection $\oab$ has the form
$$ {\rm ev}_0^*A+ \left[{\rm ev}_1^*(A-\ovA)-{\rm
    ev}_0^*(A-\ovA)\right].$$ The crucial point is that this
depends only on the end points, i.e., if $\tilde\gamma\in\pap$ and
${\tilde V}\in T_{\tilde\gamma}\pap$ then $\oab({\tilde V})$
depends only on ${\tilde V}(0)$ and ${\tilde V}(1)$. If the
conditions on $\Phi$ in (i) hold then reparametrization has the
effect of replacing each ${\tilde\Gamma}_s$ with
${\tilde\Gamma}_{\Phi^0(s)}\circ\Phi_s$, which is in $\pap$, and
the vector field $t\mapsto
\partial_{s}({\tilde\Gamma}_{\Phi^0(s)}\circ\Phi_s(t))$ is an
$\oab$-horizontal vector, because its end point values are those of 
$t\mapsto \partial_{s}({\tilde\Gamma}_{\Phi^0(s)}(t))$, since
$\Phi_s(t)$ equals $t$ if $t$ is $0$ or $1$. 

Now suppose (\ref{modiffake}) holds. Then $\oab$ becomes simply
${\rm ev}_0^*A$. In this case $\oab({\tilde V})$ depends on $\tilde
V$ only through the initial value ${\tilde V}(0)$. Thus, the
$\oab$-parallel-transport of ${\tilde\gamma}\in\pap$, along a path
$s\mapsto \Gamma_s\in\mpm$, is obtained by
$A$-parallel-transporting the initial point ${\tilde\gamma}(0)$
along the path $s\mapsto\Gamma^0(s)$, and shooting off
$\ovA$-horizontal paths lying above the paths $\Gamma_s$. (Since
the paths $\Gamma_s$ do not necessarily have the second component
fixed, their horizontal lifts need not be of the form
${\tilde\Gamma}_s\circ\Phi_s$, except at $s=0$ and $s=1$, when the
composition ${\tilde\Gamma}_{\Phi_s}\circ\Phi_s$ is guaranteed to
be meaningful.) From this it is clear that parallel translating
${\tilde\Gamma}_0\circ{\Phi_0}$, by $\oab$ along the path
$s\mapsto\Gamma_s$, results, at $s=1$, in the path
${\tilde\Gamma}_1\circ\Phi_1$.
\end{proof}

\smallskip

\underline{\em The curvature of $\omega_{(A,B)}$}
\smallskip

We can compute the curvature of the connection $\omega_{(A,B)}$.
This is, by definition,
$$\Omega_{(A,B)}=d\oab+\frac{1}{2}[\oab\wedge\oab],$$
where the exterior differential $d$ is understood in a natural
sense that will become clearer in the proof below. More
technically, we are using here notions of calculus on smooth
spaces; see, for instance, Stacey \cite{Stacey} for a survey, and
Viro \cite{Viro} for another approach.

First we describe some notation about Chen integrals in the present
context. If $B$ is a $2$-form on $P$, with values in a Lie algebra,
then its Chen integral $\int_0^1B$, restricted to $\pap$, is a
$1$-form on $\pap$ given on the vector ${\tilde V}\in
T_{{\tilde\gamma}}\bigl(\pap\bigr)$ by
$$\left(\int_0^1B\right)({\tilde
  V})=\int_0^1B\bigl({\tilde\gamma}'(t), {\tilde V}(t)\bigr)\,dt.$$
If $C$ is also a $2$-form on $P$ with values in the same Lie
algebra, we have a product $2$-form on the path space $\pap$ given
on ${\tilde X}, {\tilde Y}\in T_{{\tilde\gamma}}\bigl(\pap\bigr)$
by
\begin{equation}\label{chenpr}\begin{split}
    \left(\int_0^1\right)^2[B{\wedge} C]({\tilde X}, {\tilde
      Y})&\\
      &\hskip -1in=\int_{0\leq u<v\leq 1}
      \left[B\bigl({\tilde\gamma}'(u), 
      {\tilde X}(u)\bigr),
      C\bigl({\tilde\gamma}'(v), {\tilde
        Y}(v)\bigr)\right]\,du\,dv\\ 
    &\hskip -.75in - {}\int_{0\leq u<v\leq 1}
    \left[C\bigl({\tilde\gamma}'(u), {\tilde X}(u)\bigr),
      B\bigl({\tilde\gamma}'(v), {\tilde
        Y}(v)\bigr)\right]\,du\,dv\\
    &\hskip -1in =\int_0^1\int_0^1\left[B\bigl({\tilde\gamma}'(u),
      {\tilde 
        X}(u)\bigr), C\bigl({\tilde\gamma}'(v), {\tilde
        Y}(v)\bigr)\right]\,du\,dv.\end{split}
\end{equation}

\begin{prop}\label{P:curvom} The curvature of $\omega_{(A,B)}$ is 
\begin{equation}\label{curvom}\begin{split}
    \Omega^{\omega_{(A,B)}}&={\rm ev}_1^*F^A +d\left(\int_0^1
      {{\tau}B}\right)\\
    &\qquad\qquad + \left[{\rm ev}_1^*A{\wedge }
      \int_0^1{{\tau}B}\right]
    +\left(\int_0^1\right)^2[{{\tau}B}{\wedge}{{\tau}B}
    ],\end{split}\end{equation} where the integrals are Chen
integrals.
\end{prop}

\begin{proof} From
$$\omega_{(A,B)}={\rm ev}_1^*A+\int_0^1{{\tau}B},$$
we have
\begin{equation} \begin{split} \Omega^{\omega_{(A,B)}}&=
    d\omega_{(A,B)}+\frac{1}{2}[\omega_{(A,B)}\wedge \omega_{(A,B)}]\\
    &= {\rm ev}_1^*dA +d\int_0^1{{\tau}B} +W,
\end{split}\end{equation}
where
\begin{equation} \begin{split} W({\tilde X}, {\tilde Y}) &=
    [\omega_{(A,B)}({\tilde X}),
    \omega_{(A,B)}({\tilde Y})]\\
    &=[{\rm ev}_1^*A({\tilde X}), {\rm ev}_1^*A({\tilde Y})]\\
    & \quad + \left[{\rm ev}_1^*A({\tilde X}), \int_0^1
      {{\tau}B}\bigl({\tilde\gamma}'(t), {\tilde
        Y}(t)\bigr)\,dt\right]
    \\
    &\qquad + \left[ \int_0^1 {{\tau}B}\bigl({\tilde\gamma}'(t),
      {\tilde
        X}(t)\bigr)\,dt, {\rm ev}_1^*A({\tilde Y})\right]\\
    &\qquad \qquad+\int_0^1\int_0^1
    \left[{{\tau}B}\bigl({\tilde\gamma}'(u), {\tilde
        X}(u)\bigr),{{\tau}B}\bigl({\tilde\gamma}'(v), {\tilde
        Y}(v)\bigr)\right]\, du\,dv\\
    &=[{\rm ev}_1^*A, {\rm ev}_1^*A]({\tilde X},{\tilde Y}) +
    \left[{\rm ev}_1^*A{\wedge} \int_0^1{{\tau}B}\right]({\tilde
      X},
    {\tilde Y})\\
    &\qquad\qquad +\left(\int_0^1\right)^2[{{\tau}B}{\wedge}
    {{\tau}B}] ({\tilde X}, {\tilde Y}).
\end{split}\end{equation}
\end{proof}

In the case $A=\ovA$, and without $\tau$, the expression for the
curvature can be expressed in terms of the `fake curvature'
$F^{\ovA}+B$. For a result of this type, for a related connection
form, see Cattaneo et al. \cite[Theorem 2.6]{Cat} have calculated a
similar formula for curvature of a related connection form.

A more detailed exploration of the fake curvature would be of
interest.

\medskip
\noindent\underline{\em Parallel-transport of horizontal paths}

As before, $A$ and $\ovA$ are connections on a principal $G$-bundle
$\pi:P\to M$, and $B$ is an $LH$-valued $\alpha$-equivariant
$2$-form on $P$ vanishing on vertical vectors. Also ${\mathcal P}X$
is the space of smooth paths $[0,1]\to X$ in a space $X$, and
$\pap$ is the space of smooth $\ovA$-horizontal paths in $P$.

Our objective now is to express parallel-transport along paths in
${\mathcal P}M$ in terms of a smooth local section of the bundle
$P\to M$:
$$\sigma:U\to P$$
where $U$ is an open set in $M$.  We will focus only on paths lying
entirely inside $U$.

The section $\sigma$ determines a section ${\tilde\sigma}$ for the
bundle $\pap\to {\mathcal P}M$: if $\gamma\in\mpm$ then $
{\tilde\sigma}(\gamma)$ is the unique $\ovA$-horizontal path in
$P$, with initial point $\sigma\bigl(\gamma(0)\bigr)$, which
projects down to $\gamma$.  Thus,
\begin{equation}\label{sigmab}
{\tilde\sigma}(\gamma)(t)=\sigma(\gamma(t)){\ova}(t),
\end{equation}
for all $t\in [0,1]$, where ${\ova}(t)\in G$ satisfies the
differential equation 
\begin{equation}\label{btd}
  {\ova}(t)^{-1}{\ova}'(t)=-{\rm Ad}\bigl({\ova}(t)^{-1}\bigr)
  \ovA\left((\sigma\circ\gamma)'(t)\right)
\end{equation}
for $t\in [0,1]$, and the initial value $\ova(0)$ is $e$.

   \begin{figure}
 \begin{center}

    \includegraphics[height=2in]{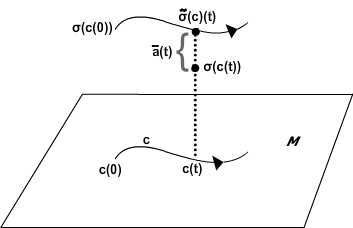}
\end{center}
 \caption{The section ${\tilde\sigma}$ applied to a path $c$}
    \label{fig:sigtildc}
   \end{figure}

Recall that a tangent vector $V\in T_{\gamma}\bigl(\mpm\bigr)$
is a smooth vector field along the path $\gamma$. Let us denote
${\tilde\sigma}(\gamma)$ by $\tilde\gamma$:
$${\tilde\gamma}\stackrel{\rm def}{=}{\tilde\sigma}(\gamma).$$
Note, for later use, that
\begin{equation}\label{tildgamprim}
{\tilde\gamma}'(t)=\sigma_*\bigl(\gamma'(t)\bigr)\ova(t)
+\underbrace{
{\tilde\gamma}(t)\ova(t)^{-1}\ova'(t)}_{\rm vertical}.
\end{equation}
Now define the vector
\begin{equation}\label{def:tildeV}
{\tilde V} ={\tilde\sigma}_*(V)\in
T_{{\tilde\gamma}}\bigl(\pap\bigr)
\end{equation}
to be the  vector ${\tilde V}$ in
$T_{{\tilde\gamma}}\bigl(\pap\bigr)$ whose 
initial value ${\tilde V}(0)$ is
$${\tilde V}(0)=\sigma_*\bigl(V(0)\bigr).$$
The existence and uniqueness of ${\tilde V}$ was proved in Lemma
\ref{l:lift}.

Note that ${\tilde V}(t)\in T_{{\tilde\gamma}(t)}P$ and
$(\sigma_*V)(t)\in T_{\sigma(\gamma(t))}P$, are generally different
vectors. However, $(\sigma_*V)(t)\ova(t)$ and ${\tilde V}(t)$ are
both in $ T_{{\tilde\gamma}(t)}P$ and differ by a vertical vector
because they have the same projection $V(t)$ under $\pi_*$:
   \begin{equation}\label{sigstarVvert}
    {\tilde V}(t)=(\sigma_*V)(t)\ova(t) + {\rm vertical\, vector}.
\end{equation}

Our objective now is to determine the $LG$-valued $1$-form
\begin{equation}\label{oaab}
\omega_{(\ovA,A,B)}={\tilde\sigma}^*\omega_{(A,B)}
\end{equation}
on $\mpm$, defined on any vector  $V\in T_{\gamma}(\mpm)$ by
\begin{equation}\label{defoaab}
\omega_{(\ovA,A,B)}(V)=
\omega_{(A,B)}\bigl({\tilde\sigma}_*V\bigr).
\end{equation}
We can now work out an explicit expression for this $1$-form.
\begin{prop}\label{P:oaab} With notation as above,
and $V\in T_{\gamma}\bigl(\mpm\bigr)$,
\begin{equation}\label{oaabv}
\omega_{(\ovA,A,B)}(V)={\rm
  Ad}\bigl(\ova(1)^{-1}\bigr)A_{\sigma}\left(V(1)\right)
+\int_0^1{\rm 
Ad}\bigl({\overline a}(t)^{-1}\bigr)
{\tau}B_{\sigma}\bigl(\gamma'(t),V(t)\bigr)\,dt,
\end{equation}
where $C_{\sigma}$ denotes the pullback $\sigma^*C$ on $M$ of a
form $C$ on $P$, and ${\overline a}:[0,1]\to G$ describes
parallel-transport along $\gamma$, i.e., satisfies
$${\overline a}(t)^{-1}{\overline a}'(t)=-{\rm Ad}\bigl({\overline
  a}(t)^{-1}\bigr)\ovA_{\sigma}\bigl(\gamma'(t)\bigr)$$ with
initial condition ${\overline a}(0)=e$. The formula for
$\omega_{(\ovA,A,B)}(V)$ can also be expressed as
\begin{equation}\label{oaabv2}\begin{split}
\omega_{(\ovA,A,B)}(V)&\\
&\hskip -.6in= A_{\sigma}\left(V(0)\right)+ \left[{\rm
    Ad}\bigl(\ova(1)^{-1}\bigr)(A_{\sigma}- 
\ovA_{\sigma})\left(V(1)\right) -(A_{\sigma}-
\ovA_{\sigma})\left(V(0)\right)\right]\\
& +\int_0^1{\rm
Ad}\bigl({\overline a}(t)^{-1}\bigr)\bigl(F^{\ovA}_{\sigma}+
{\tau}B_{\sigma}\bigr)\bigl(\gamma'(t),V(t)\bigr)
\,dt.\end{split}
\end{equation}
\end{prop}
Note that in (\ref{oaabv2}), the terms involving $\ovA_{\sigma}$
and $F^{\ovA}_{\sigma}$ cancel each other out. 

\begin{proof}
 From   the definition of $\omega_{(A,B)}$ in
(\ref{def:omegaAB}) and (\ref{def:Zv}),  we see that we need only
focus on the $B$ term. To this end we have, from
(\ref{tildgamprim}) and (\ref{sigstarVvert}): 
\begin{equation}\label{BtildegamVtild}
\begin{split}
B\bigl({\tilde\gamma'}(t),{\tilde V}(t)\bigr)  &=
B\bigl(\sigma_*\bigl(\gamma'(t)\bigr)\ova(t)
+{\rm vertical},(\sigma_*V)(t)\ova(t)+{\rm vertical} \bigr)\\
&=B\bigl(\sigma_*\bigl(\gamma'(t)\bigr)\ova(t),
(\sigma_*V)(t)\ova(t)\bigr)\\
&=\alpha\bigl(\ova(t)^{-1}\bigr)
B_{\sigma}\bigl(\gamma'(t),V(t)\bigr).
\end{split}\end{equation}
Now recall the relation (\ref{pfid})
$$\hbox{$\tau\bigl(\alpha(g)h\bigr)=g\tau(h)g^{-1}$, for all $g\in 
  G$ and $h\in H$,}$$ which implies
$$\hbox{$\tau\bigl(\alpha(g)K\bigr)={\rm Ad}(g)\tau(K)$ for all
  $g\in G$ and $K\in LH$.}$$ As usual, we are denoting the
derivatives of $\tau$ and $\alpha$ by $\tau$ and $\alpha$
again. Applying this to (\ref{BtildegamVtild}) we have
$$\tau  B\bigl({\tilde\gamma'}(t),{\tilde V}(t)\bigr) = {\rm
  Ad}\bigl(\ova(t)^{-1}\bigr){\tau} 
B_{\sigma}\bigl(\gamma'(t),V(t)\bigr),$$
and this yields  the result.
\end{proof}

Suppose $${\tilde\Gamma}:[0,1]^2\to
P:(t,s)\mapsto\tlG(t,s)=\tlG_s(t)=\tlG^t(s)$$ is smooth, with each
$\tlG_s$ being $\ovA$-horizontal, and the path $s\mapsto\tlG(0,s)$
being $A$-horizontal. Let $\Gamma=\pi\circ\tlG$.  We will need to
use the {\em bi-holonomy} $g(t,s)$ which is specified as follows:
parallel translate $\tlG(0,0)$ along $\Gamma_0|[0,t]$ by $\ovA$,
then up the path $\Gamma^t|[0,s]$ by $A$, back along
$\Gamma_s$-reversed by $\ovA$ and then down $\Gamma^0|[0,s]$ by
$A$; then the resulting point is
\begin{equation}\label{defbihol}
\tlG(0,0)g(t,s).
\end{equation}

The path $$s\mapsto\tlG_s$$ describes parallel transport of the
initial path $\tlG_0$ using the connection ${\rm ev}_0^*A$. In what
follows we will compare this with the
path $$s\mapsto{\hat\Gamma}_s$$ which is the parallel transport of
${\hat\Gamma}_0={\tilde\Gamma}_0$ using the connection ${\rm
  ev}_1^*A$. The following result describes the `difference'
between these two connections.

\begin{prop}\label{P:ptev1a} 
  Suppose $${\tilde\Gamma}:[0,1]^2\to
  P:(t,s)\mapsto\tlG(t,s)=\tlG_s(t)=\tlG^t(s)$$ is smooth, with
  each $\tlG_s$ being $\ovA$-horizontal, and the path
  $s\mapsto\tlG(0,s)$ being $A$-horizontal. Then the parallel
  translate of $\tlG_0$ by the connection ${\rm ev}_1^*A$ along the
  path $[0,s]\to \mpm:u\mapsto\Gamma_u$, where
  $\Gamma=\pi\circ\tlG$, results in $\tlG_s g(1,s)$, with $g(1,s)$
  being the `bi-holonomy' specified as in (\ref{defbihol}).
\end{prop}
\begin{proof} 
  Let ${\hat\Gamma}_s$ be the parallel translate of $\tlG_0$ by
  ${\rm ev}_1^*A$ along the path
  $[0,s]\to\mpm:u\mapsto\Gamma_u$. Then the right end point
  ${\hat\Gamma}_s(1)$ traces out an $A$-horizontal path, starting
  at $\tlG_0(1)$. Thus, ${\hat\Gamma}_s(1)$ is the result of
  parallel transporting $\tlG(0,0)$ by $\ovA$ along $\Gamma_0$ then
  up the path $\Gamma^1|[0,s]$ by $A$. If we then parallel
  transport ${\hat\Gamma}_s(1)$ back by $\ovA$ along
  $\Gamma_s|[0,1]$-reversed then we obtain the initial point
  ${\hat\Gamma}_s(0)$. This point is of the form ${\tlG}_s(0)b$,
  for some $b\in G$, and so
$${\hat\Gamma}_s={\tilde\Gamma}_sb.$$ Then, parallel-transporting
${\hat\Gamma}_s(0)$ back down $\Gamma^0|[0,s]$-reversed, by $A$,
produces the point $\tlG(0,0)b$. This shows that $b$ is the
bi-holonomy $g(1,s)$. 
\end{proof}

Now we can turn to determining the parallel-transport process by
the connection $\oab$. With $\tlG$ as above, let now
${\check\Gamma}_s$ be the $\oab$-parallel-translate of $\tlG_0$
along $[0,s]\to\mpm:u\mapsto\Gamma_u$.  Since ${\check\Gamma}_s$
and ${\tlG}_s$ are both $\ovA$-horizontal and project by $\pi_*$
down to $\Gamma_s$, we have
$${\check\Gamma}_s={\hat\Gamma}_sb_s,$$
for some $b_s\in G$. Since $\oab={\rm ev}_1^*A+\tau(Z)$ applied to
the $s$-derivative of ${\check\Gamma}_s$ is $0$, and ${\rm
  ev}_1^*A$ applied to the $s$-derivative of ${\hat\Gamma}_s$ is
$0$, we have
\begin{equation}
b_s^{-1}\partial_sb_s+{\rm
  Ad}(b_s^{-1}){\tau}Z(\partial_s{\hat\Gamma}_s)=0 
\end{equation}
Thus, $s\mapsto b_s$ describes parallel transport by
$\theta^{\sigma}$ where the section $\sigma$ satisfies $\sigma\circ
\Gamma={\hat\Gamma}$.

Since ${\hat\Gamma}_s={\tlG}_sg(1,s)$, we then have
\begin{equation}\begin{split}
    \frac{db_s}{ds}b_s^{-1}& =-{\rm
      Ad}\bigl(g(1,s)^{-1}\bigr){\tau}Z(\partial_s\tlG_s)\\ 
    &=-{\rm Ad}\bigl(g(1,s)^{-1}\bigr)\int_0^1{\tau}B\bigl(
\partial_t\tlG(t,s),\partial_s\tlG(t,s)\bigr)
\,dt\end{split}\end{equation}
To summarize:
\begin{theorem}\label{T:gb} Suppose $${\tilde\Gamma}:[0,1]^2\to
  P:(t,s)\mapsto\tlG(t,s)=\tlG_s(t)=\tlG^t(s)$$ is smooth, with
  each $\tlG_s$ being $\ovA$-horizontal, and the path
  $s\mapsto\tlG(0,s)$ being $A$-horizontal.  Then the parallel
  translate of $\tlG_0$ by the connection $\oab$ along the path
  $[0,s]\to \mpm:u\mapsto\Gamma_u$, where $\Gamma=\pi\circ\tlG$,
  results in
\begin{equation}\label{E:tilgbihol}
  \tlG_s g(1,s)\tau\bigl(h_0(s)\bigr),
\end{equation} 
with $g(1,s)$ being the `bi-holonomy' specified as in
(\ref{defbihol}), and $s\mapsto h_0(s)\in H$ solving the
differential equation
\begin{equation}\label{difeqnhs}
\frac{dh_0(s)}{ds}h_0(s)^{-1}=-{
  \alpha}\bigl(g(1,s)^{-1}\bigr)\int_0^1B\bigl( 
\partial_t\tlG(t,s),\partial_s\tlG(t,s)\bigr)
\,dt 
\end{equation}
with initial condition $h_0(0)$ being the identity in $H$.
\end{theorem}

Let $\sigma$ be a smooth section of the bundle $P\to M$ in a
neighborhood of $\Gamma([0,1]^2)$.

Let $a_t(s)\in G$ specify parallel transport by $A$ up the path
$[0,s]\to M :v\mapsto \Gamma(t,v)$, i.e.,  the
$A$-parallel-translate of $\sigma\Gamma(t,0)$ up the path $[0,s]\to
M :v\mapsto \Gamma(t,v)$ results in $\sigma(\Gamma(t,s))a_t(s)$.

On the other hand, $\ova_s(t)$ will specify parallel transport by
$\ovA$ along $[0,t]\to M:u\mapsto\Gamma(u,s)$. Thus,
\begin{equation}\label{E:tilgtssig}
{\tilde \Gamma}(t,s)=\sigma\bigl(\Gamma(t,s)\bigr)a_0(s){\ova}_s(t)
\end{equation}

The bi-holonomy is given by
$$g(1,s)=a_0(s)^{-1}\ova_s(1)^{-1}a_1(s)\ova_0(1).$$

Let us look at parallel-transport along the path
$s\mapsto\Gamma_s$, by the connection $\omega_{(A,B)}$, in terms of
the trivialization $\sigma$. Let $ {\hat\Gamma}_s\in\pap$ be
obtained by parallel transporting
${\tilde\Gamma}_0={\tilde\sigma}(\Gamma_0)\in\pap$ along the path
$$[0,s]\to M: u\mapsto\Gamma^0(u)=\Gamma(0,u).$$

This transport is described through a map
$$[0,1]\to G:s\mapsto c(s),$$
specified through
\begin{equation}\label{tildGamssig}
{\hat\Gamma}_s={\tilde\sigma}(\Gamma_s)c(s)
={\tilde\Gamma}_sa_0(s)^{-1}c(s).
\end{equation}
Then  $c(0)=e$ and
\begin{equation}\label{csp}
c(s)^{-1}c'(s)=-{\rm
Ad}\bigl(c(s)^{-1}\bigr)\omega_{(\ovA,A,B)}\bigl(V(s)\bigr),
\end{equation}
where $V_s\in T_{\Gamma_s}\mpm$ is the vector field along
$\Gamma_s$ given by
$$ V_s(t)=V(s,t)=\partial_s\Gamma(t,s)\qquad\hbox{for all $t\in
[0,1]$.}$$ Equation (\ref{csp}), written out in more detail, is
\begin{equation}\label{cps}
\begin{split}
  c(s)^{-1}c'(s)&= -{\rm Ad}\bigl(c(s)^{-1}\bigr)\Bigl[
  {\rm Ad}\bigl(\ova_s(1)^{-1}\bigr)A_{\sigma}\bigl(V_s(1)\bigr) \\
  &\hskip 1.25in{}+ \int_0^1 {\rm Ad}\bigl({\overline
    a}_s(t)^{-1}\bigr){\tau}B_{\sigma}(\Gamma_s'(t),
  V_s(t)\bigr)\,dt \Bigr],\end{split}
\end{equation} 
where ${\overline a}_s(t)\in G$ describes
$\ovA_{\sigma}$-parallel-transport along $\Gamma_s|[0,t]$.  By
(\ref{E:tilgbihol}),
$c(s)$ is given by
$$c(s)=a_0(s)g(1,s)\tau(h_0(s)),$$
where $s\mapsto h_0(s)$ solves
\begin{equation}\label{defeqnhs2}
  \frac{dh_0(s)}{ds}h_0(s)^{-1}=- \int_0^1
  \alpha\bigl({\overline
    a}_s(t)a_0(s)g(1,s)\bigr)^{-1}B_{\sigma}\bigl( 
\partial_t\Gamma(t,s),\partial_s\Gamma(t,s)\bigr)
\,dt, 
\end{equation}
with initial condition $h_0(0)$ being the identity in $H$. The
geometric meaning of ${\overline a}_s(t)a_0(s)$ is that it
describes parallel-transport first by $A_{\sigma}$ up from $(0,0)$
to $(0,s)$ and then to the right by $\ovA_{\sigma}$ from $(0,s)$ to
$(t,s)$.

\section{Two categories from plaquettes}\label{S:CatPlaq}

In this section we introduce two categories motivated by the
differential geometric framework we have discussed in the preceding
sections. We show that the geometric framework naturally connects
with certain category theoretic structures introduced by Ehresmann
\cite{Ehr1, Ehr2} and developed further by Kelley and Street
\cite{KS}.

We work with the pair of Lie groups $G$ and $H$, along with maps
$\tau$ and $\alpha$ satisfying (\ref{pfid}), and construct two
categories.  These categories will have the same set of objects,
and also the same set of morphisms.

The set of objects is simply the group $G$:
$${\bf Obj}=  G.$$   The set of morphisms
is
$${\bf Mor}=  G^4\times H,$$
with a typical element denoted $$(a,b,c,d; h).$$
It is convenient to visualize a morphism as a plaquette labeled
with elements of $G$:

\begin{figure}[ht]
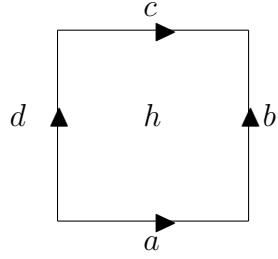

\begin{center}

\setlength{\unitlength}{.5in}

 \plaq{a}{b}{c}{d}{h}
 \end{center}
    \caption{Plaquette}
    \label{fig:Plaquette}
\end{figure}

To connect with the theory of the preceding sections, we should
think of $a$ and $c$ as giving $\ovA$-parallel-transports, $d$ and
$b$ as $A$-parallel-transports, and $h$ should be thought of as
corresponding to $h_0(1)$ of Theorem \ref{T:gb}.  However, this is
only a rough guide; we shall return to this matter later in this
section.

For the category {\bf Vert}, the   {\em source} (domain) and {\em
target} (co-domain) of a morphism are:

\begin{eqnarray*}
s_{\bf Vert}(a,b,c,d;h)
&=& a\\
t_{\bf Vert}(a,b,c,d;h) &=&c\end{eqnarray*}

For the category {\bf Horz}

\begin{eqnarray*}
s_{\bf Horz}(a,b,c,d;h)
&=&d\\
t_{\bf Horz}(a,b,c,d;h) &=&b\end{eqnarray*}

We define vertical composition, that is composition in {\bf Vert},
using Figure \ref{fig:vertcomp}. In this figure, the upper morphism
is being applied first and then the lower.

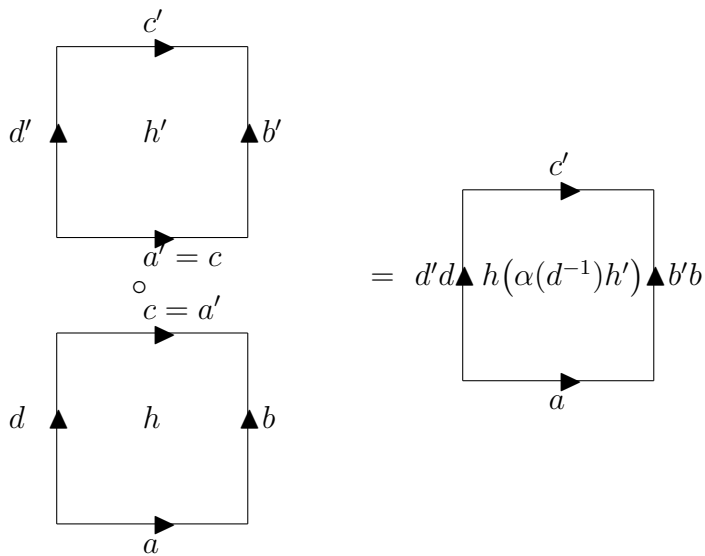
\begin{figure}[ht]
\begin{center}

\setlength{\unitlength}{.5in}

\begin{picture}(6.5,6)(0.5,1)
  \put(-1,1){\plaq{a}{b}{c=a'}{d}{h}} \put(.5,3.4){$\circ$}
  \put(3,3.5){=} \put(3.25,2.5){\plaq{a}{b'b}{c'}{d'd}{\hskip -.35in
      h\bigl(\alpha(d^{-1})h'\bigr) }}
\put(-1,4){\plaq{a'=c}{b'}{c'}{d'}{h'}}

\end{picture}
\end{center}
    \caption{Vertical Composition}
    \label{fig:vertcomp}
\end{figure}

Horizontal composition is specified through Figure
\ref{fig:horzcomp}. In this figure we have used the notation
$\circ_{\rm opp}$ to stress that, as morphisms, it is the one to
the left which is applied first and then the one to the right.

\begin{figure}[ht]
\begin{center}

\setlength{\unitlength}{.5in}

\begin{picture}(6.5,6)(2,1)
\put(-1,1){\plaq{a}{  b}{c}{d}{h}} \put(2.7,2){$\circ_{\rm opp}$}

\put(3.2,1){\plaq{a'}{b'}{c'}{d'}{h'}}

\put(6.4,2){=} \put(6.5,1){\plaq{a'a}{b'}{c'c}{d}{\hskip -.35in
\bigl(\alpha(a^{-1})h'\bigr)h}}

\end{picture}
\end{center}
    \caption{Horizontal Composition (for  $b=d'$).}
    \label{fig:horzcomp}
\end{figure}
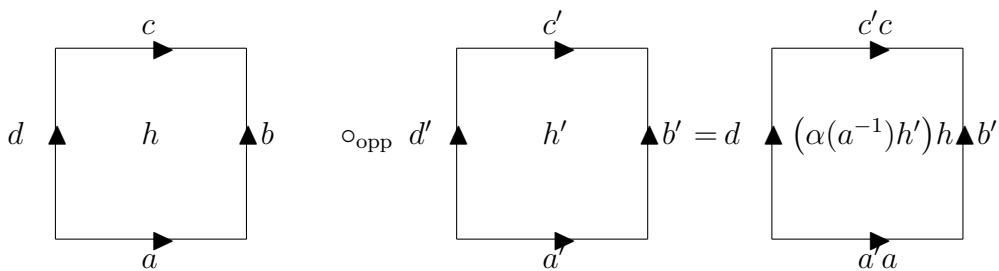

 Our first observation is:

 \begin{prop}\label{p:cat}
   Both {\bf Vert} and {\bf Horz} are categories, under the
   specified composition laws. In both categories, all morphisms are
   invertible.
 \end{prop}
 \begin{proof}It is straightforward to verify that the composition
   laws are associative. The identity map $a\to a$ in {\bf Vert} is
   $(a,e,a,e;e)$, and in {\bf Horz} it is $(e,a,e,a; e)$. These are
   displayed in in Figure \ref{fig:id}.  The inverse of the
   morphism $(a,b,c,d;h)$ in ${\bf Vert}$ is $(c,b^{-1},a,d^{-1};
   \alpha(d)h^{-1})$; the inverse in {\bf Horz} is
   $(a^{-1},d,c^{-1},b;\alpha(a)h^{-1})$. \end{proof}

 The two categories are isomorphic, but it is best not to identify
 them.

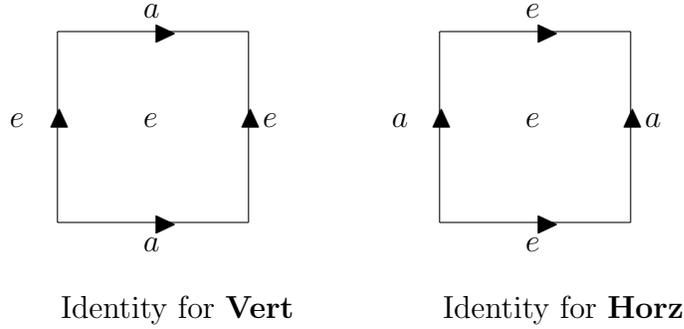
\begin{figure}[ht]
\begin{center}

\setlength{\unitlength}{.5in}

\begin{picture}(6.5,6)(2,1)
\put(1,2){\plaq{a}{  e}{a}{e}{e}} \put(1.75,1) {\mbox{Identity for
{\bf Vert}}}

\put(5,2){\plaq{e}{a}{e}{a}{  e}} \put(5.75,1) {\mbox{Identity for
{\bf Horz}}}
\end{picture}
\end{center}
    \caption{Identity Maps.}
    \label{fig:id}
\end{figure}

We use $\circ_H$ to denote horizontal composition,
and $\circ_V$ to denote vertical composition.

We have seen earlier that if $A$, $\ovA$ and $B$ are such that
$\oab$ reduces to ${\rm ev}_0^*A$ (for example, if $A=\ovA$ and
$F^\ovA+\tau(B)$ is $0$) then all plaquettes $(a,b,c,d;h)$ arising
from the connections $A$ and $\omega_{(A,B)}$,
satisfy $$\tau(h)=a^{-1}b^{-1}cd.$$ Motivated by this observation,
we could consider those morphisms $(a,b,c,d;h)$ which satisfy
\begin{equation}\label{strictquasiflat}
\tau(h)=a^{-1}b^{-1}cd
\end{equation}
However, we can look at a broader class of morphisms as
well. Suppose
$$\underbar{h}\mapsto z({\underbar{h}})\in Z(G)$$
is a mapping of the morphisms in the category {\bf Horz} or in {\bf 
  Vert} into the center $Z(G)$ of $G$, which carries composition of
morphisms to products in $Z(G)$:
$$z(\underbar{h}\circ\underbar{h}')=
z(\underbar{h})z(\underbar{h}').
$$
Then we say that a morphism $\underbar{h}=(a,b,c,d;h)$ is    {\em
quasi-flat} with respect to $z$ if
\begin{equation}\label{quasiflat}
\tau(h)=(a^{-1}b^{-1}cd) z({\underbar{h}})
\end{equation}

A larger class of morphisms could also be considered, by replacing
$Z(G)$ by an abelian normal subgroup, but we shall not explore this
here.

\begin{prop}\label{P:quasiflat} Composition of quasi-flat morphisms
  is quasi-flat.  Thus, the quasi-flat morphisms form a subcategory
  in both {\bf Horz} and {\bf Vert}.
\end{prop}
\begin{proof} Let $\underbar{h}=(a,b,c,d;h)$ and
  $\underbar{h}'=(a',b',c',d';h')$ be quasi-flat morphisms in {\bf
    Horz}, 
such that the horizontal composition $
 \underbar{h}'\circ_{H}\underbar{h}$ is defined, i.e., $b=d'$. Then
$$
 \underbar{h}'\circ_{H}\underbar{h}  =
 (a'a,b',c'c,d;\{\alpha(a^{-1})h'\}h).$$
 Applying $\tau$ to the last component in this, we have
\begin{equation}\begin{split}
a^{-1}\tau(h')a\tau(h) &=a^{-1}({a'}^{-1}{b'}^{-1}c'd')
 a(a^{-1}b^{-1}cd) z({\underbar{h}})z({\underbar{h}'})
 \\
 &= \bigl((a'a)^{-1}{b'}^{-1}(c'c)d\bigr)
 z({\underbar{h}}'\circ_H
 {\underbar{h}}),
 \end{split}
 \end{equation}
 which says that $\underbar{h}'\circ_{H}\underbar{h} $ is
 quasi-flat. 

  Now suppose $\underbar{h}=(a,b,c,d;h)$ and
  $\underbar{h}'=(a',b',c',d';h')$ are quasi-flat morphisms in {\bf 
    Vert}, 
such that the vertical composition $
 \underbar{h}'\circ_{V}\underbar{h}$ is defined, i.e., $c=a'$. Then
$$
 \underbar{h}'\circ_{V}\underbar{h}  =
 (a,b'b,c',d'd;h\{\alpha(d^{-1})h'\}).$$
 Applying $\tau$ to the last component in this, we have
\begin{equation}\begin{split}
\tau(h)d^{-1}\tau(h')d
&=(a^{-1}b^{-1}cd)d^{-1}({a'}^{-1}{b'}^{-1}c'd') 
 d  z({\underbar{h}})z({\underbar{h}'})
 \\
 &= \bigl({a'}^{-1}(b'b)^{-1}c'd'd\bigr)
 z({\underbar{h}}'\circ_V
 {\underbar{h}}),
 \end{split}
 \end{equation}
 which says that $\underbar{h}'\circ_{V}\underbar{h} $ is
 quasi-flat. 
\end{proof}

For a morphism $\underbar{h}=(a,b,c,d;h)$ we set
$$\tau(\underbar{h})=\tau(h).$$
If $\underbar{h}=(a,b,c,d;h)$ and $\underbar{h}'=(a',b',c',d';h')$
are morphisms then 
we say that they are $\tau$-equivalent,
$$\underbar{h}=_{\tau}\underbar{h}'$$
if $a=a'$, $b=b'$., $c=c'$, $d=d'$, and $\tau(h)=\tau(h')$.

\begin{prop}\label{p:hv} If
 $\underbar{h},\underbar{h}',\underbar{h}'',\underbar{h}''$ are
 quasi-flat morphisms for which the 
compositions on both sides of (\ref{hvcomp}) are meaningful, then
\begin{equation}\label{hvcomp}
(\underbar{\rm h}'''\circ_H\underbar{\rm h}'')
\circ_V(\underbar{\rm h}'\circ_H\underbar{\rm h})=_{\tau}
(\underbar{\rm h}'''\circ_V\underbar{\rm h}')\circ_H
(\underbar{\rm h}''\circ_V\underbar{\rm h}) \end{equation}
whenever all the compositions on both sides are meaningful.
\end{prop}
Thus, the structures we are using here correspond to double
categories as described by Kelly and Street \cite[section 1.1]{KS}

\begin{proof} This is a lengthy but straight forward verification.
  We refer to Figure \ref{fig:window}. For a morphism
  $\underbar{h}=(a,b,c,d;h)$, let us write 
$$\tau_{\partial}(\underbar{h})=a^{-1}b^{-1}cd.$$
For the left side of (\ref{hvcomp} ), we have
\begin{equation}\begin{split}
(\underbar{h}'\circ_H\underbar{h})
&=(a'a,b',c'c,d;\{\alpha(a^{-1})h'\}h)\\ 
(\underbar{h}'''\circ_H\underbar{h}'')
&= (c'c,b'',f'f,d'; \{\alpha(c^{-1})h'''\}h'')\\
\underbar{h}^*\stackrel{\rm def}{=}
(\underbar{h}'''\circ_H\underbar{h}'')
\circ_V(\underbar{h}'\circ_H\underbar{h})&=
(a'a,b''b',f'f,d'd; h^*),
\end{split}\end{equation}
where
\begin{equation}\label{hupstar}
h^* = \{\alpha(a^{-1})h'\}h\{\alpha(d^{-1}c^{-1})h'''\}
\{\alpha(d^{-1})h''\}\end{equation}
Applying $\tau$ gives
\begin{equation}\label{tauupstar}
\begin{split}
\tau(h^*) &= a^{-1}\tau(h')z(\underbar{h'})a\cdot
\tau(h)z(\underbar{h})d^{-1}c^{-1}\tau(h''')cd\cdot\\
&\hskip 2in
z(\underbar{h}''')\cdot d^{-1}\tau(h'')dz(\underbar{h}'')\\ 
&=(a'a)^{-1}(b''b')^{-1}(f'f)(d'd)z(\underbar{h}^*),
\end{split}\end{equation}
where we have used the fact, from (\ref{pfid}), that $\alpha$ is
converted to a conjugation on applying $\tau$, and the last line
follows after algebraic simplification. Thus, 
\begin{equation}\label{tauhupstar}\tau(h^*)=\tau_{
\partial}(\underbar{h}^*)z(\underbar{h}^*)\end{equation}
On the other hand, by an entirely similar computation, we obtain 
\begin{equation}
\underbar{h}_*\stackrel{\rm
  def}{=}(\underbar{h}'''\circ_V\underbar{h}')\circ_H 
(\underbar{h}''\circ_V\underbar{h})=
(a'a,b''b',f'f,d'd; h_*),
\end{equation}
where
\begin{equation}\label{hlowerstar}
h_* = \{\alpha(a^{-1})h'\} \{\alpha(a^{-1}b^{-1})h'''\}h
\{\alpha(d^{-1})h''\}\end{equation}
Applying $\tau$ to this yields, after using (\ref{pfid}) and
computation, 
$$\tau(h_*)=\tau_{\partial}(\underbar{h}_*)z(\underbar{h}_*) 
$$
Since $\tau(h_*)$ is equal to $\tau(h^*)$, the  result
(\ref{hvcomp}) follows. 
\end{proof}

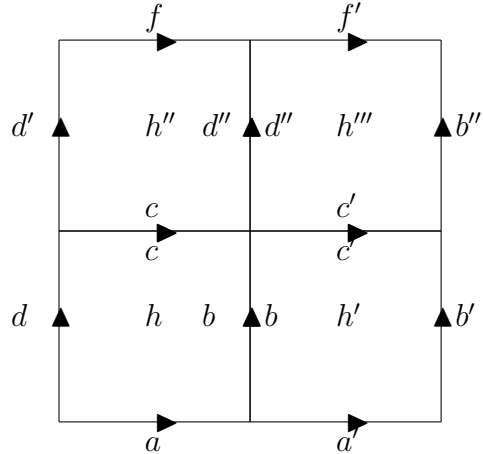
\begin{figure}[ht]
\begin{center}

\setlength{\unitlength}{.5in}

\begin{picture}(6.5,6)(2,1)
\put(2,3){\plaq{c}{  d''}{f}{d'}{h''}}


\put(4,3){\plaq{c'}{b''}{f'}{d''}{h'''}}
\put(2,1){\plaq{a}{  b}{c}{d}{h}}

\put(4,1){\plaq{a'}{b'}{c'}{b}{h'}}

\end{picture}
\end{center}
    \caption{Consistency of Horizontal and Vertical Compositions.}
    \label{fig:window}
\end{figure}

Ideally, a discrete model would be the exact `integrated' version
of the differential geometric connection $\oab$. However, it is not
clear if such an ideal transcription is feasible for any such
connection $\oab$ on the path-space bundle. To make contact with
the differential picture we have developed in earlier sections, we
should compare quasi-flat morphisms with parallel translation by
$\oab$ in the case where $B$ is such that $\oab$ reduces to ${\rm
  ev}_0^*A$ (for instance, if $A=\ovA$ and the fake curvature
$F^\ovA+\tau(B)$ vanishes); more precisely, the $h$ for quasi-flat
morphisms (taking all $z({\rm h})$ to be the identity) corresponds
to the quantity $h_0(1)$ specified through the differential
equation (\ref{difeqnhs}). It would be desirable to have a more
thorough relationship between the discrete structures and the
differential geometric constructions, even in the case when
$z(\cdot)$ is not the identity. We hope to address this in future
work.

\section{Concluding Remarks}\label{S:Conc}

We have constructed in (\ref{def:omegaAB}) a connection $\oab$ from
a connection $A$ on a principal $G$-bundle $P$ over $M$, and a
$2$-form $B$ taking values in the Lie algebra of a second structure
group $H$.  The connection $\oab$ lives on a bundle of
$\ovA$-horizontal paths, where $\ovA$ is another connection on $P$
which may be viewed as governing the gauge theoretic interaction
along each curve. Associated to each path $s\mapsto \Gamma_s$ of
paths, beginning with an initial path $\Gamma_0$ and ending in a
final path $\Gamma_1$ in $M$, is a parallel transport process by
the connection $\oab$. We have studied conditions (in Theorem
\ref{T:reparm}) under which this transport is `surface-determined',
that is, depends more on the surface $\Gamma$ swept out by the path
of paths than on the specific parametrization, given by $\Gamma$,
of this surface. We also described connections over the path space
of $M$ with values in the Lie algebra $LH$ obtained from the $\ovA$
and $B$.  We developed an `integrated' version, or a discrete
version, of this theory, which is most conveniently formulated in
terms of categories of quadrilateral diagrams. These diagrams, or
morphisms, arise from parallel transport by $\oab$ when $B$ has a
special form which makes the parallel transports
surface-determined.

Our results and constructions extend a body of literature ranging
from differential geometric investigations to category theoretic
ones.  We have developed both aspects, clarifying their
relationship.

\medskip

{\bf Acknowledgments} We are grateful to the anonymous referee for
useful comments and pointing us to the reference \cite{KS}.  Our
thanks to Urs Schreiber for the reference \cite{Stacey}.  We also
thank Swarnamoyee Priyajee Gupta for preparing some of the
figures. ANS acknowledges research supported from US NSF grant
DMS-0601141. AL acknowledges research support from Department of
Science and Technology, India under Project
No.~SR/S2/HEP-0006/2008.


\end{document}